\documentclass[reqno]{amsart}
\usepackage{
    ask,
    fullpage,
    booktabs,
    tikz-cd,
    todonotes,
    subcaption
    }
\usepackage
    [
        thmcounter=section,
        excounter=section,
    ]
    {asktheorems}

\usepackage[sorting=nyt,sortcites]{biblatex}
\usepackage{enumitem}

\numberwithin{equation}{section}

\renewcommand{\rho}{\varrho}

\DeclareMathOperator{\Proj}{Proj}

\title{On Contextuality as a Feature of Logic and Probability Theory}
\author{Ask Ellingsen}
\date{\today}

\begin{document}

\begin{abstract}
    In quantum mechanics, not everything that can be observed can be observed simultaneously. Observational data exhibits \emph{contextuality} -- a generalisation of nonlocality -- if the result of an observation is necessarily dependent on which combination of observables was measured. This article gives a mathematical introduction to contextuality, emphasising its nature as a general feature of probability theory and logic, rather than of any particular quantum theory.
\end{abstract}

\maketitle


\section{Introduction}

One hundred years after the inception of quantum mechanics (QM), there still appears to be no universal consensus on what really makes a physical system \emph{quantum} as opposed to \emph{classical}. Certainly, some physical properties are hallmarks of QM -- such as energy quantisation, uncertainty principles, and wave-particle duality -- but none of these is universally present in all ``quantum'' systems.\footnote
    {
        Indeed, not all quantum systems support notions of ``energy'', ``waves'', or ``particles'' at all. For example, where are the waves, the particles, or the energy in an (abstract) quantum bit?
    }
Conversely, there seems to be no property which is uniquely classical, in the sense that it never holds for a ``quantum'' system. Rather, quantum mechanics regularly produces predictions that are compatible with some alternative ``classical'' explanation.\footnote
    {
        Indeed, if quantum mechanics did not behave deceivingly similarly to classical mechanics ``most of the time'', then classical physics would probably never have been very useful to begin with.
    }

However, if one were to attempt to single some quintessentially quantum feature -- one of the best candidates might be the \emph{simultaneous availability of information}. Namely, in a general quantum system, \emph{not all information that may in principle be retrieved from the system may be retrieved at the same time}. If we conduct an experiment where we observe a certain feature of the system (say particle position, or spin in the $z$-direction), then there are other features (particle momentum, spin in the $x$-direction) that are \emph{unobservable} in that experiment. This is reflected in the theory -- in standard formulations of QM, no predictions are made about the outcomes of incommensurable observations (represented in the standard Hilbert space formalism by non-commuting operators).

This feature opens up an interesting possibility: The outcome of a quantum observation may in theory depend on the \emph{context} in which it is observed. By a ``context'' is meant some subset of the set of all observations that one could in principle make. The key fact is that while each \emph{individual} observation is possible, only certain \emph{combinations} of observations are possible \emph{simultaneously}.\footnote
    {
        It should be noted that unlike contextuality, the \emph{uncertainty principle} does not fundamentally require this incompatibility. Namely, the statistical variances of multiple random variables can be large even when those random variables can be simultaneously evaluated.
    }

In a classical system, there is no theoretical restriction on which contexts are possible. Therefore, we remain assured that even if we choose to only ever observe certain combinations of features in conjunction with each other, whatever information we ignore is still well-defined, whether we check in on it or not. An observation made about a classical system simply reveals information that is ``already there'', and choosing a context merely means to ignore some of this information.

By contrast, if only certain subsets of information are simultaneously accessible even in theory, it is not obvious that these subsets of information ``glue'' to a whole, predetermined truth which is \emph{independent} of the choice of which subset to access. A scenario of the second kind, where the context necessarily influences the outcome, is called \emph{contextual}. A non-contextual theory is one where a measurement can be thought of as merely revealing pre-existing information. 

Models of classical mechanics are usually non-contextual by construction. On the other hand, contextuality -- and the closely related concept of \emph{nonlocality} -- have been generally accepted features of QM (or at least of the mathematical models thereof) since at least the work of Bell, Kochen and Specker in the 1960's \autocite{bell_epr,bell_problem_1966,kochen_problem_1990}. Since then, multiple mathematical frameworks that describe some notion of contextuality have been proposed (see \autocite{budroni_kochen-specker_2022} for a fairly recent review of the current state of the art), but as of yet no one notion commands universal acclaim.

Having a readily available notion of contextuality general enough to be applicable in a multitude of scenarios seems desirable, as it would allow direct comparisons between notions of contextual behaviour in different formulations of quantum (and non-quantum) systems. Indeed, several existing notions of contextuality already apply beyond QM. Bell-type contextuality is a \emph{probabilistic} obstruction to consistent, global value assignments; and Kochen-Specker type contextuality is a \emph{logical} obstruction to the existence of any global value assignments at all.

This article is primarily intended as a review article attempting to introduce contextuality from this perspective -- as a property of collections of probability distributions rather than of QM. In particular, I present no new results of my own, but in Section \ref{sec:outlook} I outline some ideas for future work.

In Section \ref{sec:alicebob}, we introduce the concept of contextuality with a thought experiment. In Section \ref{sec:qprob}, we review how probability distributions are calculated in QM, and show how these distributions have the potential to be contextual. In Section \ref{sec:order}, we outline the theory of \emph{partial Boolean algebras}, and how these can be used to model the logical relationships between events in a contextual theory. We finish that section by giving a version of the Kochen-Specker theorem, which prohibits the existence of global value assignments for all high-dimensional enough quantum systems.

Some more modern approaches to contextuality employ the language of categories, sheaves, and topoi; and risk appearing hopelessly abstract to the uninitiated. In Section \ref{sec:outlook}, we briefly mention two of the more prominent of these approaches: The sheaf-theoretical theory of \emph{measurement scenarios} due to Abramsky and Brandenburger \autocite{abramskySheafTheoreticStructureNonLocality2011,abramsky_cohomology_2012,abramsky_contextuality_2015,abramsky_fraction_2017,abramsky_logic_2021,barbosa_continuous-variable_2022} (which quickly found applications in the theory of quantum computation, see e.g.~\autocite{raussendorf2013}); and the related \emph{Bohrification} programme, which grew out of the topos-theoretical formulation of the Kochen-Specker theorem originally proposed by Butterfield and Isham \autocite{isham_KS_I,isham_KS_II,isham_KS_III,isham_KS_IV,what_is_a_thing_2011,van_den_berg_noncommutativity_2012,berg_extending_2014,heunen_bohrification_2010,heunen_gelfand_2011}. A secondary goal of this article is to show that these formulations are well motivated by some natural modelling assumptions about contextual scenarios.
\section{Alice and Bob Receive Letters}
\label{sec:alicebob}

\begin{remark}
    This section is largely based on \autocite{abramsky_contextuality_2015}, with some slight modifications -- including making the thought experiment more ``analogue'' by phrasing it in terms of envelopes and pieces of paper; as well as emphasising the role of signalling distributions with an explicit example.
\end{remark}

To illustrate how context-dependence may appear, let us perform a simple thought experiment in the form of a game. Two players, Alice and Bob, are situated in different rooms and unable to communicate. A referee, Nate,\footnote
    {
        For \emph{Nature}.
    }
presents Alice with two envelopes, marked A0 and A1, and Bob with envelopes marked B0 and B1. Alice and Bob are each allowed to choose one, and only one, of their envelopes to open. Inside each envelope is a piece of paper marked either $0$ or $1$. This procedure is repeated many times.

For each round of the game, Alice and Bob record their choice of envelope and the number they found inside. After the experiment is over, they are allowed to compare notes. Alice and Bob's goal is to decide if Nate placed the numbers inside the envelopes before Alice and Bob made their choices, or if he tampered with them afterward.

In each round of the experiment, there are $2^4 = 16$ possible \textbf{joint outcomes}, each labelled by a quadruple $(x,y,a,b)$, where $a,b \in \{0,1\}$ represent Alice and Bob's respective choices of envelopes, and $x,y \in \{0,1\}$ the numbers they find inside. We call each envelope pair $(a,b) \in \{0,1\}^2$ a \textbf{context}, and each pair $(x,y) \in \{0,1\}^2$ a \textbf{result}.

Let $\PPP(x,y,a,b)$ be the total statistical probability of the outcome $(x,y,a,b)$, and let
\begin{equation}
    \PPP(a,b) = \sum_{x,y} \PPP(x,y,a,b)
\end{equation}
be the total probability that $a$ and $b$ were opened. We assume that $\PPP(a,b) > 0$ for all $a,b$ (no context was left unexplored). Let
\begin{equation} \label{eq:condprob}
    \PPP(x,y \mid a,b) = \frac{\PPP(x,y,a,b)}{\PPP(a,b)}
\end{equation}
be the probability of the result $(x,y)$, conditioned on the context $(a,b)$.

With this notation in place, Alice and Bob's goal is to determine whether one of the following conditions fail:
\begin{enumerate}
    \item \emph{(No tampering)} The envelopes are loaded before Alice and Bob make their choices. I.e.,~there is a well-defined \emph{hidden probability distribution} $\PPP^h(x,x',y,y')$, describing the probability that Alice's two envelopes are loaded beforehand with the numbers $x,x'$ respectively, and Bob's with $y,y'$;
    
    \item \emph{(Independence)} Alice and Bob's choices are not influenced by the preparation of the envelopes. I.e.,~the probability $\PPP(x,y,a,b)$ is given by a product
    \begin{equation} \label{eq:independence}
        \PPP(x,y,a,b) = \PPP^h|_{a,b}(x,y)\PPP(a,b),
    \end{equation}
    where
    \begin{equation}
        \PPP^h|_{0,0}(x,y) = \sum_{x',y'} \PPP^h(x,x',y,y')
    \end{equation}
    is the \emph{marginal probability} that the envelopes A0 and B0 contain the results $(x,y)$, ignoring the contents of A1 and B1 -- and similarly for all other envelope combinations $(a,b)$.
\end{enumerate}

Combining the independence assumption \eqref{eq:independence} with \eqref{eq:condprob} yields the condition
\begin{equation} \label{eq:hidden}
    \PPP(x,y \mid a,b) = \PPP^h|_{a,b}(x,y).
\end{equation}
That is, the probability to see the result $(x,y)$ \emph{conditioned} on the context $(a,b)$ is equal to the \emph{marginal} probability that envelopes $(a,b)$ were loaded with $(x,y)$ to begin with. If \emph{no} distribution $\PPP^h$ satisfying \eqref{eq:hidden} exists, then Alice and Bob can be confident that there is some correlation between their choices of envelopes and the contents therein. That is, the result is dependent on the context, or \textbf{contextual}. If \emph{some} $\PPP^h$ satisfying \eqref{eq:hidden} exists, then Alice and Bob have no evidence of tampering.

\begin{definition}
    Following \autocite{abramsky_sheaf-theoretic_2011}, with notation as above, we will refer to an array of conditionals $\PPP(x,y \mid a,b)$ as a \textbf{measurement scenario}. A probability distribution $\PPP^h : \{0,1\}^4 \to [0,1]$ satisfying \eqref{eq:hidden} will be called a \textbf{hidden global distribution} for the scenario. If a hidden distribution exists, the scenario is \textbf{non-contextual}; and otherwise \textbf{contextual}.
\end{definition}


\begin{remark}
    If a hidden distribution $\PPP^h$ satisfying \eqref{eq:hidden} exists, it is typically not unique, since a probability distribution cannot in general be reconstructed from its marginals.
\end{remark}

The 16 conditional probabilities $\PPP(x,y \mid a,b)$ can be organised into a $4\times4$ table.
These tables can be neatly visualised through the use of \emph{bundle diagrams}, first introduced in \autocite{abramsky_contextuality_2015}. To draw a bundle diagram, one draws a \emph{base graph} consisting of a vertex for each envelope A0, A1, B0, B1; and an edge connecting those vertices that belong to a common context. Over each vertex lies a \emph{fibre} consisting of two points corresponding to the results $0,1$. A joint outcome $(x,y,a,b)$ corresponds to a line connecting points of neighbouring fibres (a ``local section'' of the bundle), and the conditional probability $\PPP(x,y \mid a,b)$ can be visualised as a ``weighting'' of the corresponding line.

\begin{table}[ht]
    \caption{Probability tables for the scenarios described in Examples \ref{ex:deterministic} -- \ref{ex:hardy}.}
    \label{tab:tables}
    \begin{subtable}{0.25\textwidth}
    \centering
    \tabcolsep=.5mm
    \begin{tabular}{l|cccc}
        & $(0,0)$ & $(0,1)$ & $(1,0)$ & $(1,1)$ \\
        \midrule
        $(0,0)$ &$0$&$0$&$0$&$1$ \\
        $(0,1)$ &$0$&$0$&$0$&$1$ \\
        $(1,0)$ &$0$&$0$&$0$&$1$ \\
        $(1,1)$ &$0$&$0$&$0$&$1$ \\
    \end{tabular}
    \caption{The deterministic scenario from Example \ref{ex:deterministic}.}
    \label{tab:deterministic}
    \end{subtable}
    \hfill
    \begin{subtable}{0.3\textwidth}
    \centering
    \tabcolsep=.5mm
    \begin{tabular}{l|cccc}
        & $(0,0)$ & $(0,1)$ & $(1,0)$ & $(1,1)$ \\
        \midrule
        $(0,0)$ &$1/2$&$0$&$0$&$1/2$ \\
        $(0,1)$ &$1/2$&$0$&$0$&$1/2$ \\
        $(1,0)$ &$1/2$&$0$&$0$&$1/2$ \\
        $(1,1)$ &$1/2$&$0$&$0$&$1/2$ \\
    \end{tabular}
    \caption{The noncontextual 50/50-scenario from Example \ref{ex:50/50}.}
    \label{tab:fifty}
    \end{subtable}
    \hfill
    \begin{subtable}{0.3\textwidth}
    \centering
    \tabcolsep=.5mm
    \begin{tabular}{l|cccc}
        & $(0,0)$ & $(0,1)$ & $(1,0)$ & $(1,1)$ \\
        \midrule
        $(0,0)$ &$0$&$0$&$0$&$1$ \\
        $(0,1)$ &$1$&$0$&$0$&$0$ \\
        $(1,0)$ &$1$&$0$&$0$&$0$ \\
        $(1,1)$ &$0$&$0$&$0$&$1$ \\
    \end{tabular}
    \caption{The signalling scenario from Example \ref{ex:signal}.}
    \label{tab:signal}
    \end{subtable}
    \hfill
    \begin{subtable}{0.3\textwidth}
    \centering
    \tabcolsep=.5mm
    \begin{tabular}{l|cccc}
        & $(0,0)$ & $(0,1)$ & $(1,0)$ & $(1,1)$ \\
        \midrule
        $(0,0)$ &$1/2$&$0$&$0$&$1/2$ \\
        $(0,1)$ &$1/2$&$0$&$0$&$1/2$ \\
        $(1,0)$ &$1/2$&$0$&$0$&$1/2$ \\
        $(1,1)$ &$0$&$1/2$&$1/2$&$0$ \\
    \end{tabular}
    \caption{The Popescu-Rohrlich box from Example \ref{ex:PRbox}.}
    \label{tab:PRbox}
    \end{subtable}
    \hfill
    \begin{subtable}{0.3\textwidth}
    \centering
    \tabcolsep=.5mm
    \begin{tabular}{l|cccc}
        & $(0,0)$ & $(0,1)$ & $(1,0)$ & $(1,1)$ \\
        \midrule
        $(0,0)$ &$1/2$&$0$&$0$&$1/2$ \\
        $(0,1)$ &$3/8$&$1/8$&$1/8$&$3/8$ \\
        $(1,0)$ &$3/8$&$1/8$&$1/8$&$3/8$ \\
        $(1,1)$ &$1/8$&$3/8$&$3/8$&$1/8$ \\
    \end{tabular}
    \caption{The CHSH scenario from Example \ref{ex:CHSH}.}
    \label{tab:CHSH}
    \end{subtable}
    \hfill
    \begin{subtable}{0.3\textwidth}
    \centering
    \tabcolsep=.5mm
    \begin{tabular}{l|cccc}
        & $(0,0)$ & $(0,1)$ & $(1,0)$ & $(1,1)$ \\
        \midrule
        $(0,0)$ &$1/2$&$0$&$0$&$1/2$ \\
        $(0,1)$ &$1/2$&$0$&$0$&$1/2$ \\
        $(1,0)$ &$1/2$&$0$&$0$&$1/2$ \\
        $(1,1)$ &$1/4$&$1/4$&$1/4$&$1/4$ \\
    \end{tabular}
    \caption{The logically contextual scenario from Example \ref{ex:hardy}.}
    \label{tab:hardy}
    \end{subtable}
    \hfill
\end{table}

\begin{figure}[ht]
    \caption{Bundle diagrams corresponding to the tables in Table \ref{tab:tables}. Darker colours and thicker lines correspond to higher probabilities. Joint outcomes that are assigned zero probability are shown as dotted lines.}
    \label{fig:bundles}
    \begin{subfigure}{0.3\textwidth}
        \centering
        \begin{tikzpicture}

\coordinate (a) at (0,0);
\coordinate (b) at (2,-.5);
\coordinate (a') at (3,.5);
\coordinate (b') at (1,1);

\coordinate (a0) at (0,2);
\coordinate (b0) at (2,1.5);
\coordinate (a'0) at (3,2.5);
\coordinate (b'0) at (1,3);

\coordinate (a1) at (0,3);
\coordinate (b1) at (2,2.5);
\coordinate (a'1) at (3,3.5);
\coordinate (b'1) at (1,4);


\draw (a) -- (b) -- (a') -- (b') -- cycle;

\draw[dotted] (a) -- (a1);
\draw[dotted] (b) -- (b1);
\draw[dotted] (a') -- (a'1);
\draw[dotted] (b') -- (b'1);


\draw[fill,red] (a) node [left] {A0} circle (2pt);
\draw[fill,blue] (b) node [right] {B0} circle (2pt);
\draw[fill,red] (a') node [right] {A1} circle (2pt);
\draw[fill,blue] (b') node [left] {B1} circle (2pt);

\draw[fill] (a0) node [left] {$0$} circle (1pt);
\draw[fill] (a1) node [left] {$1$} circle (1pt);
\draw[fill] (b0) circle (1pt);
\draw[fill] (b1) circle (1pt);
\draw[fill] (a'0) circle (1pt);
\draw[fill] (a'1) circle (1pt);
\draw[fill] (b'0) circle (1pt);
\draw[fill] (b'1) circle (1pt);


\draw[dotted] (a0) -- (b0) -- (a'0) -- (b'0) -- cycle;
\draw[dotted] (a1) -- (b1) -- (a'1) -- (b'1) -- cycle;
\draw[dotted] (a0) -- (b1) -- (a'0) -- (b'1) -- cycle;
\draw[dotted] (a1) -- (b0) -- (a'1) -- (b'0) -- cycle;


\draw[very thick,DarkGreen] (a0) -- (b0) -- (a'0) -- (b'0) -- cycle;
        \end{tikzpicture}
        \caption{Bundle diagram corresponding to Table \ref{tab:deterministic}.}
        \label{fig:deterministic}
    \end{subfigure}
    \hfill
    \begin{subfigure}{0.3\textwidth}
        \centering
        \begin{tikzpicture}

\coordinate (a) at (0,0);
\coordinate (b) at (2,-.5);
\coordinate (a') at (3,.5);
\coordinate (b') at (1,1);

\coordinate (a0) at (0,2);
\coordinate (b0) at (2,1.5);
\coordinate (a'0) at (3,2.5);
\coordinate (b'0) at (1,3);

\coordinate (a1) at (0,3);
\coordinate (b1) at (2,2.5);
\coordinate (a'1) at (3,3.5);
\coordinate (b'1) at (1,4);


\draw (a) -- (b) -- (a') -- (b') -- cycle;

\draw[dotted] (a) -- (a1);
\draw[dotted] (b) -- (b1);
\draw[dotted] (a') -- (a'1);
\draw[dotted] (b') -- (b'1);


\draw[fill,red] (a) node [left] {A0} circle (2pt);
\draw[fill,blue] (b) node [right] {B0} circle (2pt);
\draw[fill,red] (a') node [right] {A1} circle (2pt);
\draw[fill,blue] (b') node [left] {B1} circle (2pt);

\draw[fill] (a0) node [left] {$0$} circle (1pt);
\draw[fill] (a1) node [left] {$1$} circle (1pt);
\draw[fill] (b0) circle (1pt);
\draw[fill] (b1) circle (1pt);
\draw[fill] (a'0) circle (1pt);
\draw[fill] (a'1) circle (1pt);
\draw[fill] (b'0) circle (1pt);
\draw[fill] (b'1) circle (1pt);


\draw[dotted] (a0) -- (b1) -- (a'0) -- (b'1) -- cycle;
\draw[dotted] (a1) -- (b0) -- (a'1) -- (b'0) -- cycle;


\draw[DarkGreen,thick] (a0) -- (b0) -- (a'0) -- (b'0) -- cycle;
\draw[DarkGreen,thick] (a1) -- (b1) -- (a'1) -- (b'1) -- cycle;
        \end{tikzpicture}
        \caption{Bundle diagram corresponding to Table \ref{tab:fifty}.}
        \label{fig:fifty}
    \end{subfigure}
    \hfill
    \begin{subfigure}{0.3\textwidth}
        \centering
        \begin{tikzpicture}

\coordinate (a) at (0,0);
\coordinate (b) at (2,-.5);
\coordinate (a') at (3,.5);
\coordinate (b') at (1,1);

\coordinate (a0) at (0,2);
\coordinate (b0) at (2,1.5);
\coordinate (a'0) at (3,2.5);
\coordinate (b'0) at (1,3);

\coordinate (a1) at (0,3);
\coordinate (b1) at (2,2.5);
\coordinate (a'1) at (3,3.5);
\coordinate (b'1) at (1,4);


\draw (a) -- (b) -- (a') -- (b') -- cycle;

\draw[dotted] (a) -- (a1);
\draw[dotted] (b) -- (b1);
\draw[dotted] (a') -- (a'1);
\draw[dotted] (b') -- (b'1);


\draw[fill,red] (a) node [left] {A0} circle (2pt);
\draw[fill,blue] (b) node [right] {B0} circle (2pt);
\draw[fill,red] (a') node [right] {A1} circle (2pt);
\draw[fill,blue] (b') node [left] {B1} circle (2pt);

\draw[fill] (a0) node [left] {$0$} circle (1pt);
\draw[fill] (a1) node [left] {$1$} circle (1pt);
\draw[fill] (b0) circle (1pt);
\draw[fill] (b1) circle (1pt);
\draw[fill] (a'0) circle (1pt);
\draw[fill] (a'1) circle (1pt);
\draw[fill] (b'0) circle (1pt);
\draw[fill] (b'1) circle (1pt);


\draw[dotted] (a0) -- (b0) -- (a'0) -- (b'0) -- cycle;
\draw[dotted] (a1) -- (b1) -- (a'1) -- (b'1) -- cycle;
\draw[dotted] (a0) -- (b1) -- (a'0) -- (b'1) -- cycle;
\draw[dotted] (a1) -- (b0) -- (a'1) -- (b'0) -- cycle;


\draw[DarkGreen,very thick] (a0) -- (b0);
\draw[DarkGreen,very thick] (a1) -- (b'1);
\draw[DarkGreen,very thick] (a'1) -- (b1);
\draw[DarkGreen,very thick] (a'0) -- (b'0);
        \end{tikzpicture}
        \caption{Bundle diagram corresponding to Table \ref{tab:signal}.}
        \label{fig:signalling}
    \end{subfigure}
    \hfill
    \begin{subfigure}{0.3\textwidth}
        \centering
        \begin{tikzpicture}

\coordinate (a) at (0,0);
\coordinate (b) at (2,-.5);
\coordinate (a') at (3,.5);
\coordinate (b') at (1,1);

\coordinate (a0) at (0,2);
\coordinate (b0) at (2,1.5);
\coordinate (a'0) at (3,2.5);
\coordinate (b'0) at (1,3);

\coordinate (a1) at (0,3);
\coordinate (b1) at (2,2.5);
\coordinate (a'1) at (3,3.5);
\coordinate (b'1) at (1,4);


\draw (a) -- (b) -- (a') -- (b') -- cycle;

\draw[dotted] (a) -- (a1);
\draw[dotted] (b) -- (b1);
\draw[dotted] (a') -- (a'1);
\draw[dotted] (b') -- (b'1);


\draw[fill,red] (a) node [left] {A0} circle (2pt);
\draw[fill,blue] (b) node [right] {B0} circle (2pt);
\draw[fill,red] (a') node [right] {A1} circle (2pt);
\draw[fill,blue] (b') node [left] {B1} circle (2pt);

\draw[fill] (a0) node [left] {$0$} circle (1pt);
\draw[fill] (a1) node [left] {$1$} circle (1pt);
\draw[fill] (b0) circle (1pt);
\draw[fill] (b1) circle (1pt);
\draw[fill] (a'0) circle (1pt);
\draw[fill] (a'1) circle (1pt);
\draw[fill] (b'0) circle (1pt);
\draw[fill] (b'1) circle (1pt);


\draw[dotted] (a0) -- (b0) -- (a'0) -- (b'0) -- cycle;
\draw[dotted] (a1) -- (b1) -- (a'1) -- (b'1) -- cycle;
\draw[dotted] (a0) -- (b1) -- (a'0) -- (b'1) -- cycle;
\draw[dotted] (a1) -- (b0) -- (a'1) -- (b'0) -- cycle;


\draw[very thick,DarkGreen] (a0) -- (b0) -- (a'0) -- (b'1) -- (a1) -- (b1) -- (a'1) -- (b'0) -- cycle;
        \end{tikzpicture}
        \caption{Bundle diagram corresponding to Table \ref{tab:PRbox}.}
        \label{fig:PRbox}
    \end{subfigure}
    \hfill
    \begin{subfigure}{0.3\textwidth}
        \centering
        \begin{tikzpicture}

\coordinate (a) at (0,0);
\coordinate (b) at (2,-.5);
\coordinate (a') at (3,.5);
\coordinate (b') at (1,1);

\coordinate (a0) at (0,2);
\coordinate (b0) at (2,1.5);
\coordinate (a'0) at (3,2.5);
\coordinate (b'0) at (1,3);

\coordinate (a1) at (0,3);
\coordinate (b1) at (2,2.5);
\coordinate (a'1) at (3,3.5);
\coordinate (b'1) at (1,4);


\draw (a) -- (b) -- (a') -- (b') -- cycle;

\draw[dotted] (a) -- (a1);
\draw[dotted] (b) -- (b1);
\draw[dotted] (a') -- (a'1);
\draw[dotted] (b') -- (b'1);


\draw[fill,red] (a) node [left] {A0} circle (2pt);
\draw[fill,blue] (b) node [right] {B0} circle (2pt);
\draw[fill,red] (a') node [right] {A1} circle (2pt);
\draw[fill,blue] (b') node [left] {B1} circle (2pt);

\draw[fill] (a0) node [left] {$0$} circle (1pt);
\draw[fill] (a1) node [left] {$1$} circle (1pt);
\draw[fill] (b0) circle (1pt);
\draw[fill] (b1) circle (1pt);
\draw[fill] (a'0) circle (1pt);
\draw[fill] (a'1) circle (1pt);
\draw[fill] (b'0) circle (1pt);
\draw[fill] (b'1) circle (1pt);


\draw[dotted] (a0) -- (b1);
\draw[dotted] (a1) -- (b0);


\draw[Green] (b0) -- (a'1) -- (b'1) -- (a0);
\draw[Green] (b1) -- (a'0) -- (b'0) -- (a1);
\draw[thick,DarkGreen!90!Green] (a0) -- (b0) -- (a'0) -- (b'1) -- (a1) -- (b1) -- (a'1) -- (b'0) -- cycle;
\draw[very thick,DarkGreen] (a0) -- (b0);
\draw[very thick,DarkGreen] (a1) -- (b1);
        \end{tikzpicture}
        \caption{Bundle diagram corresponding to Table \ref{tab:CHSH}.}
        \label{fig:CHSH}
    \end{subfigure}
    \hfill
    \begin{subfigure}{0.3\textwidth}
        \centering
        \begin{tikzpicture}

\coordinate (a) at (0,0);
\coordinate (b) at (2,-.5);
\coordinate (a') at (3,.5);
\coordinate (b') at (1,1);

\coordinate (a0) at (0,2);
\coordinate (b0) at (2,1.5);
\coordinate (a'0) at (3,2.5);
\coordinate (b'0) at (1,3);

\coordinate (a1) at (0,3);
\coordinate (b1) at (2,2.5);
\coordinate (a'1) at (3,3.5);
\coordinate (b'1) at (1,4);


\draw (a) -- (b) -- (a') -- (b') -- cycle;

\draw[dotted] (a) -- (a1);
\draw[dotted] (b) -- (b1);
\draw[dotted] (a') -- (a'1);
\draw[dotted] (b') -- (b'1);


\draw[fill,red] (a) node [left] {A0} circle (2pt);
\draw[fill,blue] (b) node [right] {B0} circle (2pt);
\draw[fill,red] (a') node [right] {A1} circle (2pt);
\draw[fill,blue] (b') node [left] {B1} circle (2pt);

\draw[fill] (a0) node [left] {$0$} circle (1pt);
\draw[fill] (a1) node [left] {$1$} circle (1pt);
\draw[fill] (b0) circle (1pt);
\draw[fill] (b1) circle (1pt);
\draw[fill] (a'0) circle (1pt);
\draw[fill] (a'1) circle (1pt);
\draw[fill] (b'0) circle (1pt);
\draw[fill] (b'1) circle (1pt);


\draw[dotted] (b'0) -- (a1) -- (b0) -- (a'1);
\draw[dotted] (b'1) -- (a0) -- (b1) -- (a'0);


\draw[very thick,DarkGreen] (b'0) -- (a0) -- (b0) -- (a'0);
\draw[very thick,DarkGreen] (b'1) -- (a1) -- (b1) -- (a'1);
\draw[Green] (a'0) -- (b'0);
\draw[Green] (a'0) -- (b'1);
\draw[Green] (a'1) -- (b'0);
\draw[Green] (a'1) -- (b'1);
        \end{tikzpicture}
        \caption{Bundle diagram corresponding to Table \ref{tab:hardy}.}
        \label{fig:hardy}
    \end{subfigure}
    \hfill
\end{figure}
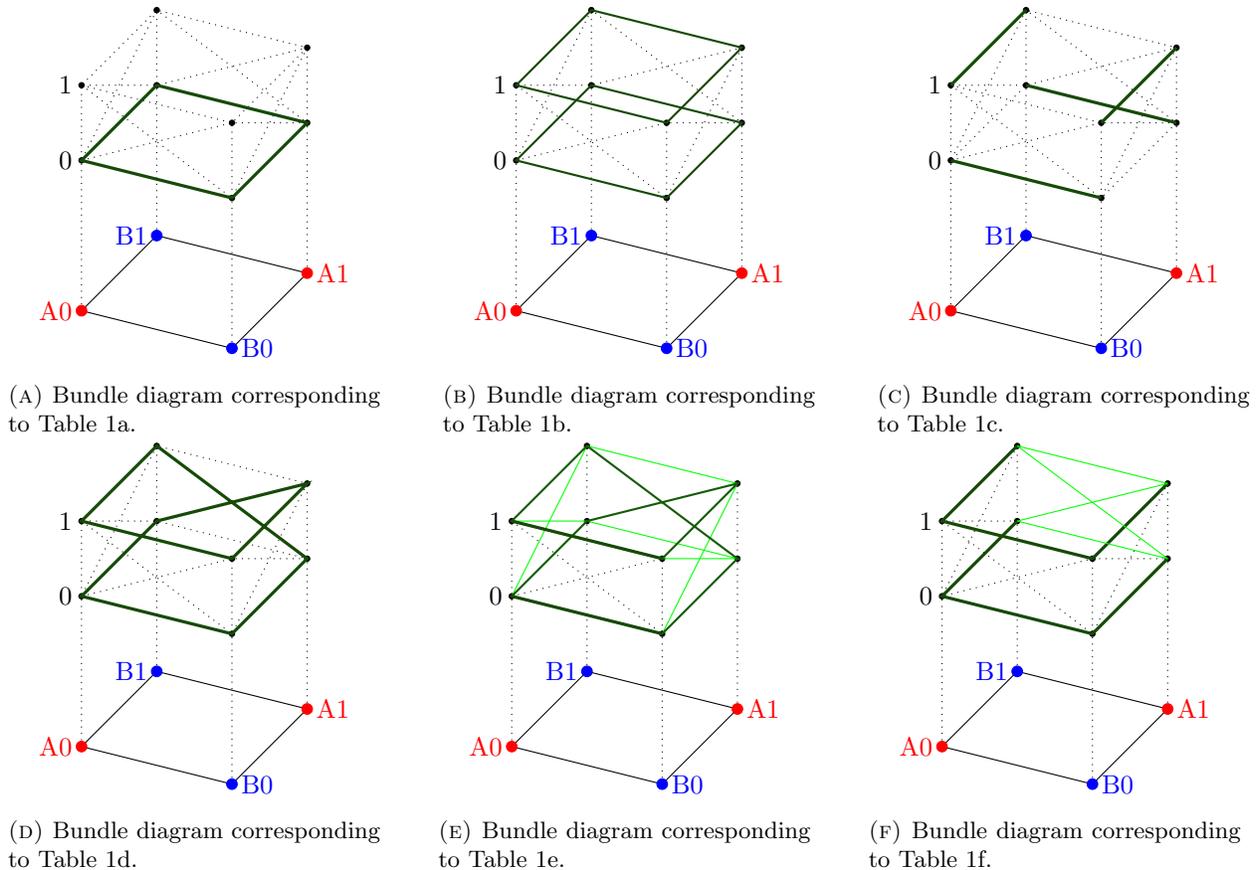

\begin{example}[Deterministic outcome] \label{ex:deterministic}
    Table \ref{tab:deterministic} (fig.~\ref{fig:deterministic}) represents a scenario where the referee always loads all envelopes with the number 1. This scenario is non-contextual.
\end{example}

\begin{example}[50/50] \label{ex:50/50}
    Table \ref{tab:fifty} (fig.~\ref{fig:fifty}) represents a scenario where the referee always loads all envelopes with the same number within a round, with a 50\% chance of choosing 0 or 1 each time. This scenario is non-contextual.
\end{example}

\begin{example}[Signalling] \label{ex:signal}
    Table \ref{tab:signal} (fig.~\ref{fig:signalling}) represents to a scenario where the referee loads Alice and Bob's envelopes with the result $1$ if they choose the context $(0,0)$ or $(1,1)$, and with $0$ if they choose $(0,1)$ or $(1,0)$. This scenario is contextual.
\end{example}

\begin{example}[Popescu-Rohrlich box] \label{ex:PRbox}
    Table \ref{tab:PRbox} (fig.~\ref{fig:PRbox}) represents a scenario where Alice and Bob always see the same result as each other (with a $50\%$ chance of each) in the contexts $(0,0)$, $(0,1)$, and $(1,0)$, but opposite results in the context $(1,1)$. This scenario is the contextual \textbf{Popescu-Rohrlich (PR) box} \autocite{popescu_1994}.
\end{example}

\begin{example}[CHSH] \label{ex:CHSH}
    Table \ref{tab:CHSH} (fig.~\ref{fig:CHSH}) represents a weakened version of the PR box from Example \ref{ex:PRbox}, where the outcomes that were possible in the PR box are the \emph{most likely}, but other outcomes are also possible. We will refer to this as is the \textbf{Clauser-Horne-Shimony-Holt (CHSH) scenario}. It is contextual.
\end{example}

We claim that Examples \ref{ex:signal} -- \ref{ex:CHSH} are all contextual, but they are so in increasingly subtle ways. In the bundle diagrams, a global assignment of results $(x,x',y,y')$ corresponds to a closed, simple loop traversing each fibre exactly once. A probability table is non-contextual precisely if the corresponding bundle diagram can be built by adding together such simple loops of constant weight \autocite{abramsky_contextuality_2015}.

Examples \ref{ex:signal} and \ref{ex:PRbox} clearly fail this condition, as they do not contain any closed loops at all. Such distributions are called \textbf{strongly contextual}. They correspond to a situation where no global assignment of values is at all compatible with the data.

In Example \ref{ex:signal}, we see that there are ``discontinuities'' in the bundle diagram. This reflects that this example is \textbf{signalling}, meaning that Alice's result allows her to deduce what envelope Bob chose to open, and vice versa. Signalling is witnessed by \emph{incompatible marginals}.

\begin{definition}
    With notation as above, a measurement scenario is \textbf{no-signalling}, or satisfies \textbf{no-disturbance} at $a$ if
    \begin{equation} \label{eq:no-signalling}
        \sum_y \PPP(x,y \mid a,0) = \sum_y \PPP(x,y \mid a,1).
    \end{equation}
    Similarly for $b$. A measurement scenario is \textbf{signalling} if the no-signalling condition fails at any $a$ or $b$.
\end{definition}

In other words, a scenario is no-signalling iff we have well-defined marginal distributions
\begin{equation}
    \PPP(x \mid a) = \sum_y \PPP(x,y \mid a,0) = \sum_y \PPP(x,y \mid a,1),
\end{equation}
and
\begin{equation}
    \PPP(y \mid b) = \sum_x \PPP(x,y \mid 0,b) = \sum_x \PPP(x,y \mid 1,b),
\end{equation}

That the presence of signalling expresses itself as a discontinuity in the bundle diagram is suggestive of the fact that this type of signalling is \emph{locally detectable}. Alice and Bob can infer information about the other's choice from their own result.

Signalling distributions are not terribly exciting, as they require Alice and Bob to be able to exchange information. What is perhaps less immediately intuitive is that there is some wiggle room between non-contextual and signalling distributions. Namely, marginalising a probability distribution loses information about correlations between events. These correlations may contain contextual information even when Alice and Bob's locally available, marginal distributions do not. The slogan is that even if the distributions are \emph{locally consistent} they may be \emph{globally inconsistent} \autocite{abramsky_contextuality_2015}. No-signalling contextual distributions can allow Alice and Bob to \emph{coordinate} without \emph{communicating}.

The PR box from Example \ref{ex:PRbox} is no-signalling -- meaning that no information is exchanged between Alice and Bob within a round; reflected in the lack of ``discontinuities'' in the bundle diagram. Nevertheless, the diagram still does not contain any closed \emph{simple} loop (only a double loop), and so is strongly contextual. To see that no local assignment can be extended to a global assignment, assume for example that A0 contains 1. This leads to the following ``liar's cycle'':
\[
    \text{\color{red} A0 contains 1} \implies \text{B0 contains 1} \implies \text{A1 contains 1} \implies \text{B1 contains 0} \implies \text{\color{red} A0 contains 0};
\]
an immediate contradiction. A similar contradiction follows from assuming that A0 contains 0, so no assignment of values to envelope A0 is consistent with \emph{any} global assignment of values to all envelopes.


If the PR box is ``qualitatively'' contextual, the CHSH scenario from Example \ref{ex:CHSH} is only \emph{quantitatively} contextual. To derive a contradiction, one can assume that envelope A0 contains 0 and calculate the conditional probability that A1 contains 0. Passing via B0 gives a conditional probability of $3/4$, while passing via B1 gives a conditional probability of only $3/8$. We say that the CHSH scenario is \textbf{probabilistically}, or \textbf{weakly contextual}.

In between strong and weak contextuality lies \textbf{logical contextuality}. In a logically (but not strongly) contextual scenario, the bundle diagram may contain \emph{some} closed simple loops, but not all edges belong to one. I.e., not every local section of the diagram is extendable to a global section.

\begin{example} \label{ex:hardy}
    Table \ref{tab:hardy} (fig.~\ref{fig:hardy}) shows a logically contextual, but not strongly contextual scenario. Namely, there \emph{exist} global sections, corresponding to the scenarios that all envelopes contain the same number (these are the two closed simple loops in fig.~\ref{fig:hardy}). But in the context $(A1,B1)$, the outcomes where Alice and Bob's envelopes contain different numbers (the ``crossbars'' in fig.~\ref{ex:hardy}) do not belong to any global section. Following the bundle diagram clockwise, we have the sequence of implications
    \[
        \text{A1 contains 0} \implies \text{B0 contains 0} \implies \text{A0 contains 0} \implies \text{B1 contains 0}.
    \]
    Similarly, if A1 contains 1, then B1 must also contain 1. Nevertheless, in the context $(A1,B1)$, half the time that A1 contains a 0, B1 contains a 1, and vice versa.
\end{example}

\begin{remark}
    The three forms of contextuality \emph{strong}, \emph{logical}, and \emph{weak} are strictly ordered by strength:
    \[
        \text{Strong contextuality} \implies \text{Logical contextuality} \implies \text{Weak contextuality}.
    \]
    However, \emph{signalling} is a somewhat independent property. A signalling distribution is always weakly contextual, but may or may not be strongly or logically contextual.
\end{remark}

The bundle diagrams from Fig.~\ref{fig:bundles} suggest that contextuality is in some way a ``topological obstruction'' to the existence of global sections. This perspective is expanded upon in \autocite{abramsky_sheaf-theoretic_2011,abramsky_cohomology_2012,abramsky_contextuality_2015}.




\section{Quantum Probability}
\label{sec:qprob}

In this section we briefly review how probabilities are calculated in the Hilbert space based textbook approach to QM. The goal of this section is to illustrate how contextual collections of probability distributions may emerge from the familiar rules of quantum mechanics. For simplicity, we will only discuss pure states and sharp measurements.

\subsection{The Born rule}

One of the main predictive laws of QM is the \emph{Born rule} \autocite{born_quantum_nodate,born_statistical_1955}. In its simplest form, the Born rule says that \emph{events} $a$ are represented by \emph{projection operators} $P_a = P_a^2 = P_a^\dagger$ acting on a Hilbert space $\cH$. A (pure) \emph{state} of the system is represented by a normalised vector $\psi\in\cH$, $\norm{\psi} = 1$; and if one performs a measurement in which an event $a$ represented by a projection $P_a$ may or may not occur, then given that the system is in the state represented by $\psi$, the probability to observe $a$ is
\begin{equation} \label{eq:born}
    \PPP_\psi(a) = \inprod{\psi}{P_a\psi}.
\end{equation}

\begin{remark}
    Notice that the Born rule \eqref{eq:born} is not a mathematical identity, but a \emph{modelling assumption}, or \emph{correspondence rule}. The quantity $\PPP_\psi(a)$ on the lhs is the \emph{modelled probability} of the event (represented by) $a$ occurring given that the system is in the state (represented by) $\psi$. The rhs is a number which can be calculated.
\end{remark}

\begin{remark}
    In the interest of simplicity, we only consider \emph{projective} measurements and \emph{pure} states. Alternatively, we may simply postulate the version of the Born rule given in \eqref{eq:born} as a toy example, and study the properties of the resulting probability distributions.
\end{remark}

The Born rule readily generalises to simultaneous measurements of multiple events $a_1,\ldots,a_n$, as long as the projection operators $P_{a_i}$ \emph{commute}. Namely, then $\prod_{i=1}^n P_{a_i}$ is a new projection operator, and the probability of simultaneously observing all the events $a_i$ is given by
\begin{equation}
    \PPP_\psi(a_1,\ldots,a_n) = \inprod{\psi}{\left(\prod_{i=1}^n P_{a_i}\right)\psi}.
\end{equation}
We may think of it as though the projection $\prod_{i=1}^n P_{a_i}$ represents the event that \emph{all} the events $a_i$ individually occur.

We will call events $a,b$ corresponding to non-commutative operators $P_a,P_b$ \textbf{incommensurable}.\footnote
    {
        In the literature on contextuality, the term \emph{incommeasurable} seems to be more common.
    }
If any of the $a_i$ are incommensurable, then not only does the product operator $\prod_{i=1}^n P_{a_i}$ ambiguously depend on the indexation of the events $a_i$; it may no longer even be a projection operator, and hence does not represent an event. The Born rule thus makes no claims about the probability that incommensurable events occur simultaneously.

In QM, one often emphasises \emph{observables} over events. These two viewpoints are equivalent, and in a sense ``duals''. An \emph{observable property} of the system is postulated to be represented by a self-adjoint operator $T = T^\dagger$, defined on a dense subspace $\cD_T$ of $\cH$. Such an operator is also called an \textbf{observable}. By the spectral theorem, each observable $T$ has a unique associated projection valued measure (PVM) $P^T : \sB(\RRR) \to \Proj(\cH)$ -- where $\sB(\RRR)$ is the $\sigma$-algebra of Borel sets on $\RRR$ and $\Proj(\cH)$ is the set of projections in $\cH$ -- such that
\begin{equation}
    T = \int_{\RRR} \lambda \, \dd P^T(\lambda).
\end{equation}
If $U \subset \RRR$ is Borel measurable, then the event ``$T$ takes a value in $U$'' -- which we will denote symbolically as $(T \leadsto U)$ -- is represented by the projection $P^T(U)$. The Born rule gives the probability of this event, assuming that the system is in a state $\psi \in \cD_T$, as
\begin{equation}
    \PPP_\psi(T \leadsto U) = \inprod{\psi}{P^T(U)\psi};
\end{equation}
and the expectation value of (the property represented by) $T$ as
\begin{equation}
    \EEE_\psi(T) = \inprod{\psi}{T\psi}.
\end{equation}
In particular, the projection $P_a$ corresponding to an event $a$ is an observable, representing the ``truth value'' of $a$. The probability that $a$ occurs agrees with the expectation value of the projection:
\[
    \PPP_\psi(a) = \inprod{\psi}{P_a\psi} = \EEE_\psi(P_a).
\]


\subsection{Contexts in QM}

Since the Born rule only assigns probabilities to sets of commuting projections, it may in principle allow contextual distributions. If events are represented by projections $P_a \in \Proj(\cH)$, then \textbf{contexts} are subsets $C \subset \Proj(\cH)$ of mutually commuting projections.

Commutativity is a reflexive and symmetric, but generally non-transitive binary relation on the set of projections. We write
\begin{equation}
    P_a \odot P_b \qq{if} P_aP_b = P_bP_a.
\end{equation}
A context is a subset $C \subset \Proj(\cH)$ such that $P_a \odot P_b$ for all $P_a,P_b \in C$.

Notice that contexts have the following property:
\begin{quote}
    \textbf{Specker's principle.} If $C \subset \Proj(\cH)$ is a set such that all $P_a,P_b\in C$ are \emph{pairwise} commensurable, then they are \emph{all} commensurable.
\end{quote}
This principle, also called \textbf{coherence}, holds because commutativity is a binary relation. In higher generality, one can imagine situations where a collection of events may be pairwise commensurable without all being commensurable as a group.

We also note that the set of projections support versions of the \emph{logical connectives} \emph{and}, \emph{or}, and \emph{not} -- but only between commensurable projection operators. If $P_a \odot P_b$, then the event \emph{$a$ and $b$} is represented by the projection $P_aP_b$; the event \emph{$a$ or $b$} is represented by $P_a + P_b - P_aP_b$; and the event \emph{not $a$} is represented by $1-P_a$.

Furthermore, the projection onto the trivial subspace 0 is always assigned probability 0 by the Born rule \eqref{eq:born}, so may be taken as representing an \emph{impossible} event; and the identity projection $1: \cH \to \cH$ always has probability 1, so it represents a \emph{necessary} event (anything happens).

Projection operators are equivalent to linear subspaces, and these subspaces are partially ordered by inclusion. In terms of projections, the inclusion order of subspaces is equivalent to the order
\begin{equation}
    P_a \leq P_b \qq{if} P_aP_b = P_bP_a = P_a.
\end{equation}
Under interpretation, if $P_a \leq P_b$, then $a$ \emph{implies} $b$.

In the classic paper \autocite{birkhoff_logic_1975}, Birkhoff and von Neumann attempted to use the set of projections to provide a logical foundation for a probability theory where probabilities are calculated using the Born rule. In their approach, dubbed (Birkhoff-von Neumann) \emph{quantum logic}, the order structure of the space of projections was assigned the highest importance.

Later, Kochen and Specker introduced another approach to quantum logic using the language of \emph{partial Boolean algebras}, which placed more emphasis on the commensurability relation \autocite{kochen_problem_1990}. This will be the topic of the next section.
\section{Logic, Probability, and Order}
\label{sec:order}

As the thought experiment involving Alice and Bob in Section \ref{sec:alicebob} illustrates, contextuality is a general feature of logical relationships and probability distributions, not only of quantum mechanics. Hence it seems natural to attempt to look for a general formulation of contextuality in terms of probability distributions, rather than directly in terms of quantum mechanical notions.

In this section we will see that the standard Kolmogorovian formulation of probability theory using measure theory on $\sigma$-algebras (or Boolean algebras) of sets is not sufficiently general to support any notion of contextuality, since it assumes the existence of a \emph{sample space} of global outcomes.

\begin{remark}
    Boole himself (who did his work on probability long before the invention of set theory) actually recognised the necessity of additional assumptions to avoid Bell-like phenomena within his framework -- see Section 11 of Chapter XIX in \autocite{boole1854}. Later authors have also recognised the necessity of certain constraints to be satisfied in order for ``local'' probability distributions to ``glue'' -- see e.g.~Vorob'ev \autocite{vorobev1962}. For an interesting discussion on the history, see the introduction to \autocite{gill2022} (where the references to Boole and Vorob'ev can also be found). I thank the referee for this comment, as I was unaware of this connection between Boole's and Bell's works.
\end{remark}

Furthermore, thanks to a famous theorem due to Stone (Theorem \ref{thm:stone}), the existence of a sample space can be inferred from the order theoretical structure describing the logical relationships between events, so in order to allow contextual probability distributions we must first allow a more general model of the events themselves. At the end of this section we explore the potential of \emph{partial Boolean algebras} as one such model.

\subsection{What is probability?}

The \emph{probability} of an event is, in the most general sense of the word, simply a number indicative of some form of \emph{confidence} in the event -- either that it will occur, has occurred, or an estimate of how often the event can be expected to occur, on average. A probability \emph{distribution} assigns probabilities to multiple events.

As a starting point for a naive mathematical model of probability, we may take a set $\cA$, whose elements represent \textbf{events}. How one goes about assigning probabilities to events in $\cA$ is a matter of modelling -- it could be through statistical testing, plain guesswork, or otherwise -- but generally, a mathematical model of probability is assumed to satisfy some version of the following assumptions:
\begin{enumerate}
    \item The probability is 0 that \emph{nothing} happens, and 1 that \emph{something} happens. That is, if $\bot \in \cA$ is the \textbf{impossible event} and $\top \in \cA$ is the \textbf{necessary event}, then
    \begin{equation} \label{eq:P-normal}
        \PPP(\bot) = 0, \qquad \PPP(\top) = 1.
    \end{equation}
    \item Probability is \emph{additive} with respect to the logical connectives \emph{and} and \emph{or}:
    \begin{equation} \label{eq:P-add}
        \PPP(a \text{ or } b) + \PPP(a \text{ and } b) = \PPP(a) + \PPP(b).
    \end{equation}
\end{enumerate}


\begin{remark}
    The above axioms are not intended as mathematically well-formed statements, only as intuitive ``sketches'' motivating the definitions making up a concrete model of probability. Do note however that the second axiom implies that the conjunctions ``$a$ and $b$'' as well as ``$a$ or $b$'' must necessarily themselves exist as \emph{single events}, so that they can be assigned probabilities.
\end{remark}

\subsection{Measure theoretic probability}

In the standard Kolmogorovian, measure theory based approach to modelling (classical) probability distributions, a core assumption is the existence of a \textbf{sample space} $\Omega$ (a set). As the set of events one takes a collection $\cA \subset \sP(\Omega)$ of subsets $a \subset \Omega$ ($\sP(\Omega)$ is the \emph{power set} of $\Omega$). A point $\omega \in \Omega$ is an \textbf{outcome}. We say that $a$ \textbf{occurs} (w.r.t.~$\omega$) if $\omega \in a$. Each outcome $\omega$ thus yields an evaluation map $\delta_\omega : \cA \to \{0,1\}$, defined by
\begin{align}
    \delta_\omega(a) = \chi_a(\omega) =
    \begin{cases}
        1, &\qqr{if} \omega \in a; \\
        0, &\qqr{if} \omega \notin a.
    \end{cases}
\end{align}
In this way, each point (outcome) $\omega$ corresponds to a \emph{global} assignment of simultaneous \emph{truth values} $\delta_\omega(a) \in \{0,1\}$ to all events $a \in \cA$.

Sample spaces mimic an experimental scenario, where after each round of the experiment the outcome $\omega$ is definite, but beforehand is unpredictable. Morally at the very least, a probability distribution assigns a fundamental probability to each outcome $\omega$, and the total probability of an event is the sum of the probabilities of all outcomes that imply it. In the limit of infinitely many outcomes, this leads to the suggestive integral formula
\begin{equation} \label{eq:integral-prob}
    \PPP(a) = \int_{a} \, \dd \PPP(\omega).
\end{equation}
Furthermore, if $a,b\in\cA$ are events, then \emph{both} $a$ and $b$ occur (w.r.t.~$\omega$) iff $\omega \in a \cap b$, and \emph{at least one of them} occur iff $\omega \in a \cup b$. Hence it makes sense to interpret $a \cap b$ as the event \emph{$a$ and $b$}, and $a \cup b$ as \emph{$a$ or $b$}. It is then immediate from \eqref{eq:integral-prob} that the additivity axiom
\begin{equation}
    \PPP(a \cup b) + \PPP(a \cap b) = \PPP(a) + \PPP(b)
\end{equation}
holds. We take $\bot = \varnothing$ as the impossible event (no outcome at all), and $\top = \Omega$ as the necessary event (any outcome); and normalise so that $\PPP(\Omega) = 1$.

A probability distribution constructed in this way can never be contextual, since it is constructed by summing over global truth value assignments. We may attempt to remedy this is by removing the sample space $\Omega$, and work directly with the space of events $\cA$.

\subsection{Boolean algebras}

For the normalisation and additivity axioms on a probability distribution to make sense, the event set $\cA$ needs to carry some additional structure modelling the logical relationships between events. One popular way to equip $\cA$ with such structure is by making it into a \emph{Boolean algebra}.

\begin{definition} \label{def:boolean-alg}
    Let $\cA$ be a partially ordered set (poset). We say that $\cA$ is a \textbf{lattice} if the meet $a \wedge b = \inf \{a,b\}$ and join $a \vee b = \sup \{a,b\}$ exist for all $a,b\in\cA$. The lattice $\cA$ is:
    \begin{itemize}
        \item
        \textbf{Distributive} if the meet and join satisfy the \textbf{distributive law}
        \begin{equation}
            a \wedge ( b \vee c ) = ( a \wedge b ) \vee ( a \wedge c ), \qquad a,b,c\in\cA.
        \end{equation}

        \item
        \textbf{Bounded} if there are elements $\top,\bot\in\cA$ such that
        \begin{equation}
            \bot \leq a \leq \top, \qquad a \in \cA;
        \end{equation}

        \item
        \textbf{Complemented} if it is bounded, and there for each $a \in \cA$ exists $\neg a \in \cA$ such that
        \begin{equation}
            a \wedge \neg a = \bot, \qquad a \vee \neg a = \top.
        \end{equation}
    \end{itemize}
    A \textbf{Boolean algebra} is a complemented, bounded, distributive lattice.
\end{definition}

We interpret $\top$ as the necessary event, $\bot$ as the impossible event, $a \wedge b$ as the event \emph{$a$ and $b$}, $a \vee b$ as \emph{$a$ or $b$}, and $\neg a$ as \emph{not $a$}. The partial ordering relation $\leq$ can be thought of as \emph{implication}.

\begin{remark}
    The partial order relation $\leq$ can be recovered from the binary operations $\wedge,\vee : \cA \times \cA \to \cA$, via
    \begin{equation}
        a \leq b \qq{if} a \wedge (a \vee b) = a.
    \end{equation}
    Therefore, a Boolean algebra can be equivalently defined as a set $\cA$ equipped with appropriate operations $(\wedge,\vee,\neg,\top,\bot)$.
\end{remark}


A \textbf{Boolean algebra morphism} is a function $f : \cA \to \cB$ between Boolean algebras $\cA$ and $\cB$ such that
\begin{align}
    f(a \wedge b) = f(a) \wedge f(b), && f(a \vee b) = f(a) \vee f(b), && f(\top_\cA) = \top_\cB, && f(\bot_\cA) = \bot_{\cB}.
\end{align}
A Boolean algebra morphism automatically also preserves complements: $f(\neg a) = \neg f(a)$. With these morphisms, Boolean algebras form a category $\BA$. Two Boolean algebras are \textbf{isomorphic} if they are isomorphic in this category. Explicitly, an \textbf{isomorphism} of Boolean algebras is an invertible Boolean algebra morphism $f$ whose inverse $f^{-1}$ is also a Boolean algebra morphism. If $\cA$ and $\cB$ are isomorphic we write $\cA \cong \cB$. A \textbf{Boolean subalgebra} of a Boolean algebra $\cA$ is a subset $\cS \subset \cA$ which is a Boolean algebra with the same $(\leq,\wedge,\vee,\neg,\top,\bot)$. A Boolean subalgebra is a subobject in the category theoretical sense.

\begin{example}
    The \textbf{trivial}, or \textbf{terminal} Boolean algebra is the one-point set $1 = \{0\}$, and the \textbf{initial} Boolean algebra is the two-point set $2 = \{0,1\}$, each equipped with the usual order.
\end{example}

\begin{example}
    For any set $\Omega$, the power set $\sP(\Omega)$ is a Boolean algebra under the inclusion order. For example, if $\Omega = \varnothing$, then $\sP(\Omega) \cong 1$ is the trivial Boolean algebra, and if $\Omega = 1 = \{0\}$ is a singleton, then $\sP(\Omega) \cong 2$.
\end{example}

\begin{definition}
    A \textbf{concrete} Boolean algebra is one that is isomorphic to a Boolean \emph{subalgebra} of some power set $\sP(\Omega)$.
\end{definition}

A concrete Boolean algebra is one that can be realised as a set of subsets. \emph{Stone's representation theorem} (Theorem \ref{thm:stone}) implies that \emph{all} Boolean algebras are actually of this type. This means that the existence of global outcomes does not have to be assumed, but is a consequence of the abstract order-theoretic definition.

As a matter of notation, we employ the convention that a natural number $n$ is identified with the set of its predecessors: $0 = \varnothing$ and $n = \{0,1,\ldots,n-1\}$. In particular, $1 = \{0\}$ is the one-point set, and $2 = \{0,1\}$ is the two-point set; and a truth valuation on $\cA$ is a map $\omega : \cA \to 2$. Also, if $\CC$ is a category and $X,Y\in\CC$, we write
\[
    \CC(X,Y) = \Hom_{\CC}(X,Y) = \{f : X \to Y\}
\]
for the set of $\CC$-morphisms from $X$ to $Y$.

\begin{lemma}
    If $\cA$ is a Boolean algebra, let $S(\cA) = \BA(\cA,2)$ be the set of Boolean algebra morphisms $\omega : \cA \to 2$.\footnote
        {
            Each $\omega \in S(\cA)$ is equivalent to an \emph{ultrafilter} on $\cA$ -- namely the preimage $\omega^{-1}(1)$.
        }
    The sets of the form $U_a = \{\omega \in S(\cA) \mid \omega(a) = 1\}$ form a basis for a topology on $S(\cA)$ as $a$ varies over $\cA$.
\end{lemma}

We call the thus defined topology the \textbf{Stone topology} on $S(\cA)$, and $S(\cA)$ the \textbf{Stone space}, or the \textbf{Stone spectrum} of the Boolean algebra $\cA$.\footnote
    {
        We will generally opt for the term \emph{Stone spectrum} when referring to the Stone space $S(\cA)$ of a particular Boolean algebra $\cA$ (see e.g.~\autocite{berg_extending_2014,van_den_berg_noncommutativity_2012}), as the term \emph{Stone space} usually refers to any compact totally disconnected Hausdorff space.
    }
The idea is that each $\omega \in \BA(\cA,2)$ forms a \emph{point} of a topological space, which can be taken as a sample space for $\cA$. Each event $a \in \cA$ can then be identified with the set $U_a$ of outcomes that assign $a$ a value of 1. In this way we recover the standard picture of a Boolean algebra as a set of sets.

\begin{theorem}[Stone] \label{thm:stone}
    If $\cA \in \BA$ is a Boolean algebra and $S(\cA)$ is its Stone spectrum, then the set $\mathrm{Cl}(S(\cA))$ of clopen sets in $S(\cA)$ is a Boolean algebra under the inclusion order, and
    \begin{equation}
        \mathrm{Cl}(S(\cA)) \cong \cA.
    \end{equation}
    Conversely, if $X$ is any Stone space (totally disconnected Hausdorff space), its set of clopens $\mathrm{Cl}(X)$ is a Boolean algebra, and
    \begin{equation}
        S(\mathrm{Cl}(X)) \cong X,
    \end{equation}
    i.e., the Stone space $S(\mathrm{Cl}(X))$ is homeomorphic to $X$.
\end{theorem}

\begin{proof}
    See \autocite{stone_rep_1936}.
\end{proof}

\begin{corollary}
    Every Boolean algebra is concrete.
\end{corollary}

Stone's theorem says that the sample space $\Omega$ is implicit in the logical structure of the event set $\cA$. As long as $\cA$ carries a partial order making it into a Boolean algebra, we may always realise it as a Boolean algebra of subsets of some sample space $\Omega$. Moreover, the construction explicitly identifies the sample space $\Omega$ with the set of global truth value assignments, or \emph{global outcomes} $\omega : \cA \to 2$.

\begin{remark}
    We choose to focus on Boolean algebras in favour of $\sigma$-algebras as it simplifies the theory, but pause to note that a similar representation theorem is true for $\sigma$-algebras: the \emph{Loomis-Sikorski theorem} \autocite{sikorski1969boolean,tao_loomis}. We will not discuss $\sigma$-algebras in any detail here, but note that the Loomis-Sikorski theorem comes with an additional conceptual caveat: While Stone's theorem says that any Boolean algebra is a \emph{subalgebra} of a power set, the Loomis-Sikorski theorem says that any $\sigma$-algebra is a \emph{quotient algebra} of a subalgebra of a power set. In particular, if $(X,\cA,\mu)$ is a measure space, modding out by sets of $\mu$-measure 0 gives an ``abstract'' $\sigma$-algebra that generally is not realisable as a $\sigma$-algebra of sets.
\end{remark}

\subsection{Partial Boolean algebras}

Stone's theorem shows that if we assume that the event set $\cA$ is equipped with a partial order making it into a Boolean algebra, then it can always be realised as a Boolean algebra of sets. But this means that events are represented by sets of global outcomes, which precludes the possibility of contextuality.

This is perhaps not too surprising from a modelling perspective, since the motivation for the definition of a Boolean algebra was precisely to guarantee that the events \emph{$a$ and $b$}, \emph{$a$ or $b$}, and \emph{not $a$} are represented whenever $a$ and $b$ are. In a situation where the status of certain combinations of events are not simultaneously decidable, this motivation falls apart.

For this reason, one approach to contextual logic is through the use of \emph{partial} Boolean algebras, originally introduced by Kochen and Specker \autocite{kochen_problem_1990}. The following formulation is adapted from \autocite{van_den_berg_noncommutativity_2012}.

\begin{definition}
    A \textbf{partial Boolean algebra} consists of a set $\cA$ together with:
    \begin{itemize}
        \item A reflexive and symmetric binary \textbf{commensurability} relation $\odot \subseteq \cA \times \cA$;
        \item Elements $\bot,\top \in \cA$;
        \item A (total) unary operation $\neg : \cA \to \cA$;
        \item (Partial) binary operations $\wedge,\vee : \odot \to \cA$;
    \end{itemize}
    such that every set $C \subset \cA$ of pairwise commensurable elements $a \odot b$ is contained in a set $C \subset C' \subset \cA$ of pairwise commensurable elements, on which the above operations and elements define a Boolean algebra structure.\footnote
        {
            In particular, we must have $\bot,\top\in C'$, and $\neg a \in C'$ for each $a \in C'$.
        }
\end{definition}

A Boolean algebra is a partial Boolean algebra where the commensurability relation $\odot$ is total. We will therefore refer to Boolean algebras also as \emph{total} Boolean algebras. A morphism of partial Boolean algebras $\cA$ and $\cB$ is a map $f : \cA \to \cB$ such that $f(a) \odot f(b)$ whenever $a \odot b$, and which preserves the operations wherever defined. Partial Boolean algebras equipped with these morphisms form a category, which we denote $\pBA$ (following \autocite{abramsky_logic_2021}).

\begin{remark}
    While a total Boolean algebra can be equivalently defined as a partially ordered set, the same is not true for \emph{partial} Boolean algebras, since the commensurability condition $\odot$ is generally not transitive.
\end{remark}

The \textbf{contexts} of a partial Boolean algebra $\cA$ are sets of mutually commensurable events. It is part of the definition that each context is included in some total Boolean subalgebra, so we generally restrict attention to contexts that are themselves total Boolean algebras. If $C \subset \cA$ is a total Boolean subalgebra, we can construct its Stone space $S(C)$. The assignment of Stone spectra to Boolean algebras defines a contravariant functor,\footnote
    {
        That $S$ is a \emph{contravariant functor} (i.e., a \emph{presheaf}) means that every Boolean algebra morphism $f : \cA \to \cB$ corresponds to a continuous function in the opposite direction, $S(f) = f^* : S(\cB) \to S(\cA)$.
    }
which we call the \textbf{Stone presheaf}
\begin{equation}
    S : \BA^{\op} \to \Stone,
\end{equation}
where $\Stone$ is the category of Stone spaces. Restricting this presheaf to the total Boolean subalgebras $C \subset \cA$ of a given partial Boolean algebra $\cA$ yields an outcome space $S(C)$ for every context $C$, but no \emph{global} outcome space for the whole set $\cA$.

\begin{example}
    If $\cH$ is a Hilbert space, the set $\Proj(\cH)$ of projections on $\cH$ is a partial Boolean algebra, where $P_a \odot P_b$ if $P_aP_b = P_bP_a$. Meets, joins, and complements are given by
    \[
        P_a \wedge P_b = P_aP_b, \qquad P_a \vee P_b = P_a + P_b - P_aP_b, \qquad \neg P_a = 1-P_a.
    \]
    The Stone construction provides an outcome space $S(C)$ for each set of commuting projections, but not for every projection at once.
\end{example}

\begin{remark}
    The definition of a partial Boolean algebra takes commensurability to be a binary relation. This means that a partial Boolean algebra by definition satisfies Specker's principle. This suggests that while partial Boolean algebras can be useful for quantum mechanics, they cannot function as a setting for formulations of more general, \emph{incoherent} theories of contextual logic and probability theory.
\end{remark}

\subsection{The Kochen-Specker Theorem}

The problem of hidden variables in quantum mechanics asks whether observations made in quantum mechanics can be explained by a ``hidden'' theory where all quantum observables are given simultaneous, context-independent values. The Kochen-Specker theorem answers this question strongly in the negatory.\footnote
    {
        This section is largely based on the papers \autocite{van_den_berg_noncommutativity_2012,berg_extending_2014}.
    }

Given a collection of partially overlapping ``contexts'' $\sC = \{C_i\}_i$, each which is a Boolean algebra, we can build from it a \emph{partial} Boolean algebra $\cA$. Namely, we have the following theorem:

\begin{theorem}[van den Berg-Heunen]
    The category $\pBA$ is complete and co-complete.
\end{theorem}

\begin{proof}
    See \autocite{van_den_berg_noncommutativity_2012}.
\end{proof}

\begin{corollary}
    If $\sC = \{C_i\}_i$ is a directed set of Boolean algebras, there exists a unique (up to isomorphism) partial Boolean algebra
    \begin{equation}
        \cA = \mathrm{colim} \, \sC.
    \end{equation}
\end{corollary}

What this means is (by definition of $\cA$ being a colimit), that if we have a collection of truth valuations (Boolean algebra morphisms) $\omega_i : C_i \to 2$, one for every context $C_i \in \sC$, and if these \emph{agree on overlaps} in the sense that
\begin{equation}
    \omega_C|_{C_i\cap C_j} = \omega_{C_j}|_{C_i \cap C_j}, \qquad C_i,C_j\in\sC;
\end{equation}
then there is a unique morphism of partial Boolean algebras $\omega : \cA \to 2$ making the diagram
\begin{equation}
\begin{tikzcd}
	& \cA \\
	{C_i} & 2
	\arrow[dashed,"\omega", from=1-2, to=2-2]
	\arrow[hook, from=2-1, to=1-2]
	\arrow["{\omega_i}"', from=2-1, to=2-2]
\end{tikzcd}
\end{equation}
commute for every $i$.

Furthermore, a partial Boolean algebra is characterised precisely as the colimit of its total Boolean subalgebras (contexts).

\begin{theorem}
    If $\cA$ is a partial Boolean algebra, and $\sC(\cA)$ is the collection of total Boolean subalgebras of $\cA$, ordered by inclusion, then
    \begin{equation}
        \cA = \mathrm{colim} \, \sC(\cA).
    \end{equation}
\end{theorem}

In particular, a global truth valuation (partial Boolean algebra morphism) $\omega : \cA \to 2$ exists if and only if there is a collection of compatible local truth valuations $\omega_i : C_i \to 2$, defined on the total Boolean subalgebras of $\cA$. The Kochen-Specker theorem says that for a sufficiently high-dimensional quantum system, \emph{no} such collection of valuations can exist.

\begin{theorem}[Kochen-Specker]
    Let $\cA = \Proj(\CCC^3)$, considered as a partial Boolean algebra. If $\cB$ is a total Boolean algebra and $\omega : \cA \to \cB$ is a partial Boolean algebra morphism, then $\cB = 1$ is the trivial Boolean algebra.
\end{theorem}

\begin{proof}
    See \autocite{berg_extending_2014,kochen_problem_1990}.
\end{proof}

\begin{corollary}
    There exists no morphism of partial Boolean algebras $\Proj(\CCC^3) \to 2$.
\end{corollary}

In other words, it is impossible to assign all events (projections) on $\Proj(\CCC^3)$ consistent truth values. As $\Proj(\CCC^3)$ is included as a subalgebra of $\Proj(\cH)$ for any higher-dimensional Hilbert space $\cH$, the problem only gets worse in higher dimension. The Kochen-Specker theorem strongly obstructs the existence of any hidden variable model for quantum mechanics.
\section{Outlook}
\label{sec:outlook}

This article can only hope to give a small glimpse into the field of contextuality. Much work has been done on the subject, and much is not yet fully understood (by anyone, and much less than that by this particular author). Let us therefore finish with a brief outlook.

We have not covered the topic of \emph{weak} contextuality beyond the illustrating example in Section \ref{sec:alicebob}. As long as we are content with using a partial Boolean algebra $\cA$ as our event set, it seems natural to generalise the notion of probability distributions from Boolean algebras to partial Boolean algebras; as a collection of probability measures $\mu_C : C \to [0,1]$ defined on total Boolean subalgebras $C$.

If $\cA$ is not embeddable into a total Boolean algebra, then there is no hope of realising the $\mu_C$ as marginals of a global distribution, as there are no global outcomes to assign probabilities in the first place. However, even if $\cA$ \emph{is} a total Boolean algebra (or is embeddable into one), it may be that the $\mu_C$ cannot be extended to a global distribution. Then the $\mu_C$ exhibit weak, but not strong contextuality.

A conceptual difference between strong and weak contextuality as loosely defined above is that if a collection of local outcomes $\omega_i : C_i \to 2$ agree on overlaps, then they can always be glued to a global truth valuation. The presence of strong contextuality, which rules out global truth valuations, hence also rules out the existence of a compatible family of \emph{local} truth valuations covering the space of events. This is a familiar scenario from topology and geometry, where local functions may always be glued to global functions if they agree on the overlaps of their domains.

However, it is possible to find local \emph{probability distributions} $\mu_i : C_i \to [0,1]$ which agree on overlaps and cover the entire space $\cA$ of events, but nevertheless do not glue to any global distribution. This is because the right notion of ``restriction'' for probability distributions is \emph{marginalisation}, which loses information about correlations. Conversely, it is possible that there is \emph{no} global set of correlations that restrict to a given set of ``local'' correlations.

In conclusion, while strong contextuality appears as the lack of a global family of compatible local sections, weak contextuality can appear as the failure of a global family of compatible local sections to glue. The slogan is that \emph{weak contextuality is the failure of the presheaf of local probability distributions to be a sheaf}.

In future work, I would like to address this idea in more detail. It is my belief that a sheaf-theoretical formulation of probability theory is not only the right arena to study contextuality in, but arguably is as much as natural and rich a framework for probability theory as it has proven to be for topology, geometry, and algebra. Possibly, much of this could already be understood by experts in point-free approaches to probability theory; but less so by the quantum mechanics community. I hope this article can be of some help in connecting these communities.

Finally, let us mention some interesting approaches to contextuality that we have not had time to discuss at any length here. As discussed in the introduction, two of the main approaches relying on the theory of sheaves and topoi to capture the idea of contextuality as ``local to global incompatibility of information'' are:
\begin{itemize}
    \item The sheaf-theoretical approach to contextuality pioneered by Abramsky and Brandenburger in \autocite{abramsky_sheaf-theoretic_2011}, de-emphasises the role of Boolean algebras, and instead studies general value assignments to sets of commensurable observables. Emphasis is placed on the potential use of cohomological methods to observe contextuality. See \autocite{abramsky_sheaf-theoretic_2011,abramsky_cohomology_2012,abramsky_contextuality_2015,abramsky_fraction_2017,barbosa_continuous-variable_2022}.

    \item The sheaf-theoretic approach of Abramsky et al was preceded, and at least in part inspired by, the topos-theoretical formulation of the Kochen-Specker theorem by Isham and Butterfield \autocite{isham_KS_I,isham_KS_II,isham_KS_III,isham_KS_IV}; from which eventually grew the ambitious programme of \emph{Bohrification}. See \autocite{heunen_bohrification_2010,doring_spectral_state,isham_topos_foundation_I,isham_topos_foundation_II,isham_KS_III,what_is_a_thing_2011}. The formulation of the Kochen-Specker theorem used in this article is based on \autocite{van_den_berg_noncommutativity_2012,berg_extending_2014}, exploiting (and contributing to) this perspective.
\end{itemize}

There is far too much ongoing research related to contextuality to list all relevant directions here, but here are at least two other honourable mentions:

\begin{itemize}
    \item Fritz and others have explored category theoretical foundations of probability theory, including the introduction of \emph{Markov categories}. The basic idea is to treat probability as a theory of stochastic processes, or ``noisy maps'', which can be modelled as arrows in an appropriate category. See \autocite{fritz_synthetic_2020,fritz_representable_2023,fritz_infinite_2020}.

    \item In the literature, one also finds the notion of \emph{Spekkens contextuality}, which is another formalisation of contextuality which can be applied also to non-quantum scenarios; and which precedes Abramsky Brandenburger contextuality by a few years. I will make no attempt to recount Spekkens' formalisation of contextuality here -- instead I refer to his quite readable original paper \autocite{spekkens2005}. It is however worth noting that despite Spekkens' definition of contextuality making no mention of sheaves, and being overall quite different in nature from Abramsky and Brandenburger's definition, Wester has shown that the two definitions coincide for a large subclass of scenarios \autocite{wester2018} (see however Example 1 of Section 6 of the above reference for an example of when they differ).
\end{itemize}

The literature relevant to contextuality is vast and difficult to overview. I hope this article can at least open a window to a small corner of it.
\section*{Acknowledgements}

I would like to thank Politecnico di Milano (Polimi) for organising the intensive period on quantum mathematics, and for inviting me to contribute to this special issue -- with special thanks to Alessandro Olgiati for his tireless organisational work. I am very grateful to Douglas Lundholm and Tore Ellingsen for thoroughly proofreading the manuscript and giving helpful suggestions, comments and corrections. I also thank Lasse Svensson, Filippa Getzner, and Ella Wiiand for additional corrections, comments, and good questions. I especially want to thank the anonymous referee for their constructive feedback, interesting remarks, and additional references that greatly improved the final version of this article. Finally, I gratefully acknowledge financial support from Polimi, from C.F.~Liljewalchs Stipendiestiftelse, and from the Center for Interdisciplinary Mathematics at Uppsala University.

\section*{Declarations and statements}

{\bf Research funding}. This research was funded in part by the Center for Interdisciplinary Mathematics (CIM) at Uppsala University. \\

{\bf Conflict of interest}. The author declares that he has no competing interests regarding the publication of this paper.\\

{\bf Availability of data}. There are no data associated with the research in this paper.\\

\printbibliography

\end{document}